\UseRawInputEncoding
\documentclass[10pt,journal,twocolumn]{IEEEtran}
\usepackage{graphicx}
\usepackage{dcolumn}
\usepackage{bm}
\usepackage[utf8]{inputenc}
\usepackage{tgheros}
\usepackage{textgreek}
\usepackage[caption=false,font=footnotesize]{subfig}
\usepackage[english]{babel}
\usepackage{booktabs}
\usepackage{amsthm}
\usepackage{amsmath}
\usepackage{amssymb}
\usepackage{cite}

\usepackage{etoolbox}
\usepackage{algorithm}
\usepackage{algorithmic}
\usepackage{makecell}
\usepackage{aliascnt}

\newtheorem{proposition}{Proposition}
\newtheorem{lemma}{Lemma}
\newtheorem{theorem}{Theorem}

\newaliascnt{apptheorem}{theorem}
\newtheorem{apptheorem}[apptheorem]{Theorem}
\aliascntresetthe{apptheorem}
\gappto\appendix{\setcounter{apptheorem}{0}}

\begin{document}

\title{The Landscape of Fairness: An Axiomatic and Predictive Framework for Network QoE Sensitivity}

\author{
    Zhiyuan Ren,~\IEEEmembership{Member,~IEEE,}
    Xinke Jian,
    Wenchi Cheng,~\IEEEmembership{Senior Member,~IEEE,}
    ~and Kun Yang,~\IEEEmembership{Fellow,~IEEE}
\thanks{This work was supported by the National Key Research and Development Program of China (No. 2024YFE0200300).}%
\thanks{Zhiyuan Ren, Xinke Jian and Wenchi Cheng are with the School of Telecommunications Engineering, Xidian University, Xi'an 710071, China.}
\thanks{Kun Yang is with the School of Computer Science and Electronic Engineering, University of Essex, Colchester CO4 3SQ, U.K.}
\thanks{Corresponding author: Zhiyuan Ren (zyren@xidian.edu.cn).}%
}

\maketitle

\begin{abstract}
Evaluating network-wide fairness is challenging because it is not a static property but one highly sensitive to Service Level Agreement (SLA) parameters. This paper introduces a complete analytical framework to transform fairness evaluation from a single-point measurement into a proactive engineering discipline centered on a predictable sensitivity landscape. Our framework is built upon a QoE-Imbalance metric whose form is not an ad-hoc choice, but is uniquely determined by a set of fundamental axioms of fairness, ensuring its theoretical soundness. To navigate the fairness landscape across the full spectrum of service demands, we first derive a closed-form \textbf{covariance rule}. This rule provides an interpretable, local compass, expressing the fairness gradient as the covariance between a path's information-theoretic importance and its parameter sensitivity. We then construct \textbf{phase diagrams} to map the global landscape, revealing critical topological features such as robust "stable belts" and high-risk "dangerous wedges". Finally, an analysis of the landscape's curvature yields actionable, topology-aware design rules, including an optimal "Threshold-First" tuning strategy. Ultimately, our framework provides the tools to map, interpret, and navigate the landscape of system sensitivity, enabling the design of more robust and resilient networks.
\end{abstract}


\section{Introduction}
\label{sec:intro}
Modern networks are tuned to meet strict Service Level Objectives (SLOs)—such as tail latency targets or hop-count caps. Yet sharpening the acceptance criterion can trigger a paradoxical failure mode: small changes in an SLA's strictness or performance threshold can cause traffic to discontinuously collapse onto a few “winner-takes-most” routes. This traffic polarization not only creates severe unfairness but also induces systemic fragility, leaving the network brittle and vulnerable to cascading failures. The core challenge for operators is that this behavior is unpredictable. What is missing is not another fairness index, but a framework to analyze and predict this \textbf{sensitivity}, turning the reactive, trial-and-error art of SLA tuning into a proactive engineering discipline.

This gap persists because existing methods are ill-equipped for such a sensitivity analysis. The literature can be broadly organized into three domains, each addressing related but distinct problems.

First, a significant body of research focuses on resource allocation algorithms designed to achieve specific, static fairness objectives. These works provide sophisticated mechanisms for enforcing either max-min fairness \cite{1, 2, 3, 4, 5, 6, 7, 8} or proportional fairness \cite{9, 10, 11} across diverse network architectures, from wireless sensor networks to emerging MEC and 5G systems. While crucial for operational implementation, the primary goal of these studies is the optimization of a network state with respect to a fixed fairness criterion. Our work takes a different perspective: instead of proposing a new allocation algorithm, we provide a framework to analyze the inherent sensitivity of any given network configuration, revealing how its fairness landscape dynamically responds to the full spectrum of tunable SLA parameters.

Second, the field of traffic engineering has explored the concept of robustness, particularly through robust routing mechanisms \cite{12, 13}. These approaches aim to maintain performance under uncertain or fluctuating traffic demands, often by defining a static "uncertainty set" and optimizing routing for the worst-case scenario within that set \cite{14}. More advanced strategies combine this proactive planning with reactive, real-time anomaly detection to reconfigure routes when unexpected events occur \cite{15, 16}. This paradigm treats large deviations as external shocks to be detected and mitigated. In contrast, our framework models the system's response to parameter changes as an intrinsic, predictable property. We do not merely react to anomalies; we map the entire sensitivity landscape to forecast precisely where these high-risk "performance cliffs" will emerge as a function of the SLA, transforming reactive traffic management into a proactive design discipline.

Finally, a third and rapidly growing domain applies machine learning and deep learning models to network anomaly detection \cite{17, 18, 19}. These data-driven techniques have proven effective at identifying unusual patterns in complex, high-dimensional data from a wide array of systems, including IP networks \cite{20}, mobile networks \cite{21}, automated vehicles \cite{22}, IoT infrastructure \cite{23}, and operational KPIs \cite{24, 25}. However, these models often function as "black boxes." They may predict that a failure will occur but cannot explain why in terms of fundamental system parameters and topology, nor can they provide operators with a navigable map of the performance landscape. They offer prediction without interpretability and lack a foundational theory for proactive design. Our framework bridges this gap by providing not only prediction but also deep, causal insights through an interpretable, physically-principled model.

While fairness metrics based on Shannon entropy are well-established, the \textbf{sensitivity landscape} of these metrics—how they dynamically respond to changes in Service Level Agreement (SLA) parameters—remains largely unexplored. The contribution of this work is therefore not to propose a new metric, but to provide the \textbf{first systematic framework to map, interpret, and navigate this landscape}. 

\begin{table*}[h]
\centering
\caption{Conceptual Comparison of Network Evaluation Frameworks}
\label{tab:conceptual_comparison}
\small
\begin{tabular}{p{2.5cm}|p{4.5cm}|p{4.5cm}|p{4.5cm}}
\toprule
\textbf{Feature Dimension} & \textbf{Traditional Static Analysis} & \textbf{Machine Learning Models} & \textbf{Our Sensitivity Landscape Framework} \\
\midrule
\textbf{Evaluation Object} & Static topological properties (e.g., centrality, degree distribution). & Point-wise numerical prediction of QoS/QoE. & The entire performance landscape under tunable SLA parameters. \\
\addlinespace
\textbf{Interpretability} & Medium (based on graph-theoretic definitions). & Very Low ("black box"). & \textbf{High} (based on physical interpretations of covariance and curvature). \\
\addlinespace
\textbf{Predictive Power} & None (descriptive, not predictive). & Strong (predicts "\emph{what}" will happen). & \textbf{Strong} (predicts "\emph{what, how, and why}" it will happen). \\
\addlinespace
\textbf{Sensitivity to SLA} & None ("context-blind"). & Indirect (SLA is an input feature). & \textbf{The core object of analysis}. \\
\addlinespace
\textbf{Guidance for Action} & Heuristic design principles. & Predictive alerts, enabling reactive responses. & \textbf{Actionable engineering rules} (e.g., "Threshold-First") for proactive design. \\
\addlinespace
\textbf{Theoretical Basis} & Graph Theory. & Statistical Learning Theory. & \textbf{Axiomatic System} + Information Theory + Calculus. \\
\bottomrule
\end{tabular}
\end{table*}

We summarize this conceptual comparison in Table~\ref{tab:conceptual_comparison}. This paper develops such a framework, centered on an axiomatically-derived QoE-Imbalance metric ($I$) that captures functional fairness through a tunable Sigmoid function ($s(h;a,h_0)$). Our contributions are threefold:

\begin{enumerate}
    \item \textbf{A Predictive Landscape Map with Theoretical Anchors:} We reframe the SLA parameter space as a "spectrum of service demands" and construct "Phase Diagrams" that serve as a global map of the fairness landscape. This map reveals critical features like robust "Stable Belts" and high-risk "Dangerous Wedges". The map's boundaries are not empirical but are anchored by two rigorously proven asymptotic laws: a small-$a$ quadratic law ($I \propto \mathrm{Var}(h)$) and a large-$a$ staircase limit.

    \item \textbf{An Interpretable, Local Navigation Compass:} To navigate this landscape, we derive a closed-form, covariance-based expression for the imbalance gradient ($\nabla I$). This "Covariance Rule" transforms the complex question of “When does tightening an SLA backfire?”  into an interpretable, one-line diagnostic. It acts as a local compass, providing principled, real-time guidance for parameter tuning.

    \item \textbf{Actionable, Curvature-Driven Design Rules:} We analyze the landscape's curvature (the Hessian matrix) to quantify risk. This analysis reveals a fundamental asymmetry in the landscape, leading to our most significant design recommendation: the "Threshold-First Rule," an optimal and robust tuning sequence for network operators.
\end{enumerate}

Ultimately, this work transforms fairness from a static score into a predictable and navigable landscape, offering a new design philosophy for engineering more fundamentally robust and resilient network systems, as Fig.~\ref{fig:overview} shows.

\begin{figure}
    \centering
    \includegraphics[width=1\linewidth]{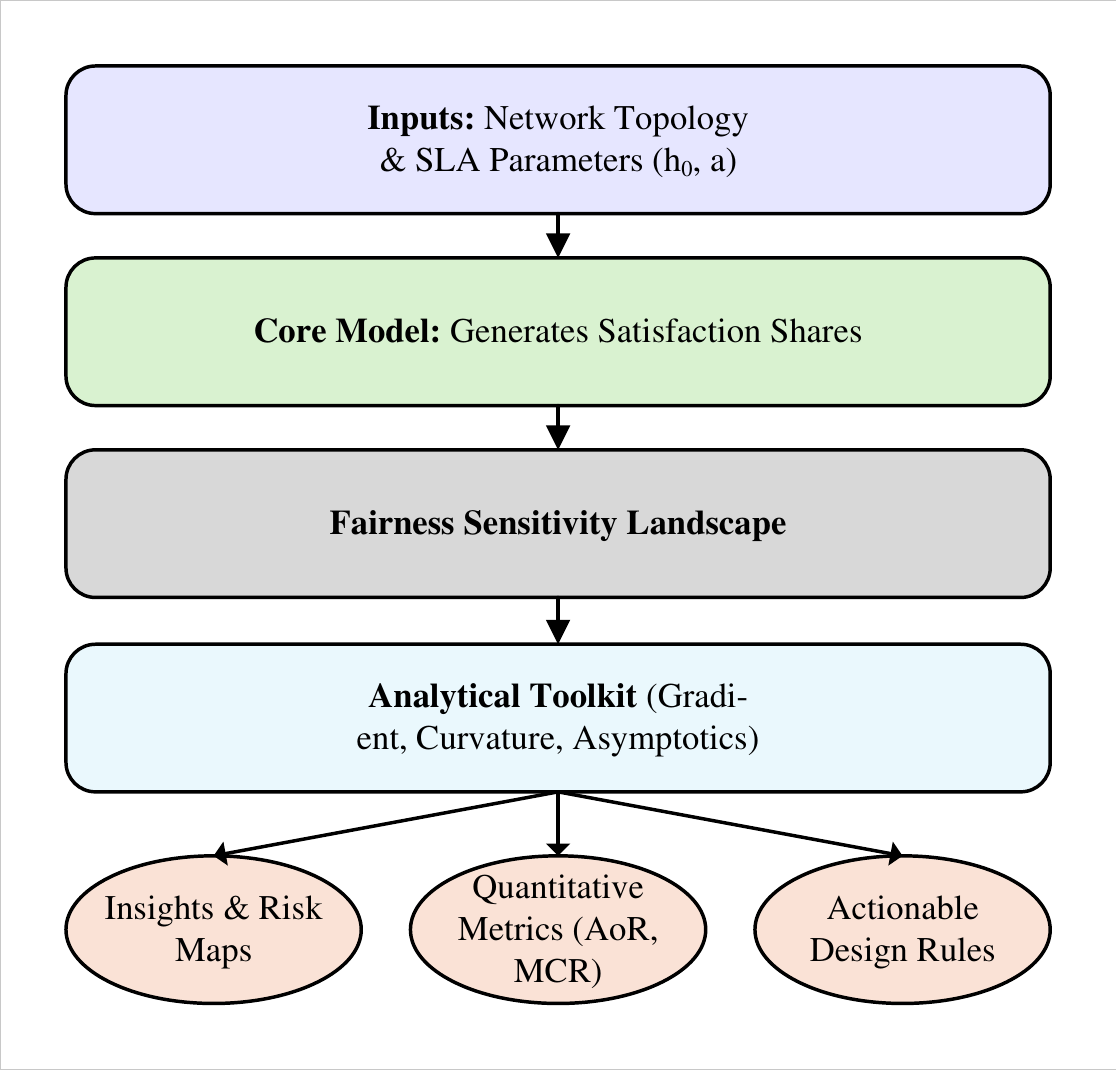}
    \caption{The conceptual workflow of the proposed fairness sensitivity analysis framework. The framework is structured as a top-to-bottom pipeline, transforming network topology and SLA parameters into a sensitivity landscape. This landscape is then analyzed by a dedicated toolkit, which simultaneously produces three types of outputs: actionable insights, quantitative metrics, and concrete design rules.}
    \label{fig:overview}
\end{figure}

The remainder of this paper is organized as follows. Section~\ref{sec:framework} formalizes our QoE-Imbalance framework, including its axiomatic foundation and analytical properties. Section~\ref{sec:sensitivity_framework} establishes our sensitivity analysis methodology and its theoretical foundations. Section~\ref{sec:results} presents the phase diagrams for a wide range of topologies and distills our key design rules. Section~\ref{sec:conclusion} concludes the paper.
\section{A Principled Framework for Measuring QoE Imbalance}
\label{sec:framework}
This section establishes the theoretical foundation of our analysis. We begin by using an axiomatic approach to derive the unique mathematical form of our QoE-Imbalance metric, $I$. We then define the SLA-aware satisfaction model that serves as the input to this metric and conclude by discussing the framework's analytical properties and computational complexity.

\subsection{An Axiomatic Foundation for Imbalance}
To rigorously quantify imbalance, we first ask: what are the fundamental properties that any rational measure of unfairness, $F$, must possess? We consider a set of $M \equiv N(N-1)$ entities (in our case, the node pairs), each with a satisfaction score $s_i>0$. From these, we define normalized satisfaction shares $p_i = s_i / \sum_j s_j$, which form a probability vector $\mathbf{p} \in \Delta^{M-1}$. We posit that any valid imbalance function $F(\mathbf{p})$ must satisfy five basic axioms:

\paragraph{Regularity Conditions.}
Throughout this section, we assume that the imbalance function $F$ satisfies standard "mild regularity" conditions common in characterizations of Shannon entropy. Specifically: (R1) $F$ is Borel-measurable on the probability simplex $\Delta^{M-1}$; and (R2) for any given number of outcomes $m \ge 2$, the bivariate function $\alpha \mapsto F(\alpha, 1-\alpha, 0, ..., 0)$ is monotonic. These conditions are sufficient to exclude pathological, non-physical solutions and ensure the uniqueness of the logarithmic form.

\begin{itemize}
    \item \textbf{A1. Anonymity:} The metric must be invariant to the reordering of shares ($F(p_1, ..., p_M) = F(p_{\sigma(1)}, ..., p_{\sigma(M)})$ for any permutation $\sigma$). It treats all entities equally, regardless of their identity.
    \item \textbf{A2. Scale Invariance:} The metric must depend only on the relative shares, not their absolute scale ($F(\mathbf{p}) = F(\lambda\mathbf{p})$ for any $\lambda>0$). For instance, measuring satisfaction in percentages or absolute scores should not change the assessed level of fairness.
    \item \textbf{A3. Calibration:} The metric must be minimized (e.g., 0) for a perfectly uniform distribution (perfect fairness) and maximized for a "winner-takes-all" scenario where one entity holds the entire share (absolute unfairness).
    \item \textbf{A4. Transfer Principle (Pigou-Dalton Condition):} A small transfer from a larger share to a smaller one (without changing their relative rank) must not increase the measured imbalance. This captures the intuitive notion of "reducing inequality."
    \item \textbf{A5. Decomposability:} The total imbalance of a system should be consistently decomposable into the imbalance \emph{between} subgroups and the weighted average of imbalances \emph{within} each subgroup. This is crucial for hierarchical or multi-level analysis.
\end{itemize}

While several metrics are used to quantify inequality, they are not all equally suited for this task from a theoretical standpoint. As shown in Table~\ref{tab:axiom_comparison}, many common indices fail to satisfy one or more of these fundamental axioms.

Notably, metrics like the Gini index, Jain's Fairness Index (JFI), and the Coefficient of Variation (CV) all fail the crucial Decomposability axiom. \footnote{The decomposability property is critical for hierarchical network analysis. For instance, a network operator might need to assess whether unfairness originates from inter-AS routing policies (between-group imbalance) or from traffic management within a single AS (within-group imbalance). A decomposable metric like ours provides a consistent framework for such multi-level diagnostics, whereas non-decomposable metrics like the Gini index cannot.}Variance, a basic measure of dispersion, fails both Scale Invariance and Decomposability. In contrast, it is a well-established theorem in information theory that the only function (up to a positive affine transformation) that simultaneously satisfies all five axioms is the Shannon entropy gap.

\begin{apptheorem}[Uniqueness of the Imbalance Metric]
\label{thm:uniqueness_final}
The only function that satisfies Axioms A1-A5 is, up to a positive affine transformation, the Shannon entropy gap: $F(\mathbf{p}) = c(\log M - H(\mathbf{p}))$. A proof sketch is provided in Appendix~\ref{app:proof_uniqueness_thm}.
\end{apptheorem}

This theorem provides a powerful justification for our choice of metric. By normalizing this gap to the range $[0, 1]$, we formally define the QoE-Imbalance $I$:
\begin{align}
    I(\mathbf{p}) &= 1 - \frac{H(\mathbf{p})}{\log_2(M)} \label{eq:imbalance_final_2} \\
    \text{where } H(\mathbf{p}) &= - \sum_{i=1}^{M} p_{i} \log_2(p_{i}) \label{eq:entropy_final_2}
\end{align}
An imbalance of $I=0$ signifies perfect fairness (all shares are equal), while $I=1$ signifies absolute unfairness (one entity holds all the share).

\begin{table}[ht!]
\centering
\caption{Comparison of Imbalance Metrics Against Fairness Axioms}
\label{tab:axiom_comparison}
\begin{tabular}{lccccc}
\toprule
\textbf{Axiom} & \textbf{$I$} & \textbf{Gini} & \textbf{JFI} & \textbf{Variance} & \textbf{CV} \\
\midrule
A1. Anonymity        & $\checkmark$ & $\checkmark$ & $\checkmark$ & $\checkmark$ & $\checkmark$ \\
A2. Scale Invariance & $\checkmark$ & $\checkmark$ & $\checkmark$ & $\times$     & $\checkmark$ \\
A3. Calibration      & $\checkmark$ & $\checkmark$ & $\checkmark$ & $\checkmark$ & $\checkmark$ \\
A4. Transfer Principle & $\checkmark$ & $\checkmark$ & $\checkmark$ & $\checkmark$ & $\checkmark$ \\
A5. Decomposability  & $\checkmark$ & $\times$     & $\times$     & $\times$     & $\times$     \\
\bottomrule
\end{tabular}
\end{table}

\subsection{SLA-Aware Satisfaction Model}

\subsubsection*{Model Scope and Rationale}
To reveal the fundamental relationship between a network's structure and its fairness sensitivity, we deliberately build our model upon a simplified, static foundation. We use hop count as the path cost metric, $h(u,v)$, and analyze the network as a static topology. This approach is a purposeful scientific abstraction designed to isolate the impact of the underlying graph structure from dynamic, operational variables like traffic congestion or link quality fluctuations. Our framework is, however, generic; the cost function $h(u,v)$ can be readily extended to incorporate more complex metrics such as latency or composite QoS scores. The present work thus establishes an essential \textbf{structural baseline}, providing a clean, topology-driven performance landscape against which all future dynamic analyses can be compared.

To compute the imbalance $I$, we first need to define the satisfaction shares $\{p_{u,v}\}$, which are derived from an underlying QoE satisfaction score, or weight, $w(u,v)$, for each node pair. This requires mapping an objective path cost—here, the hop count $h(u,v)$ between nodes $u, v \in V$—to a subjective QoE score.

To ground this mapping in established telecommunication principles, we adopt the standardized approach recommended by the ITU-T in recommendations such as G.1030. This suggests modeling the non-linear relationship between technical parameters and user satisfaction using a logistic curve, which is a canonical form of the Sigmoid function. This model accurately captures the saturation of experience at performance extremes and the existence of a critical threshold where perceived quality changes most drastically. We thus define the satisfaction score using a tunable Sigmoid function with two SLA parameters: the evaluation strictness $a > 0$ and the performance threshold $h_0 > 0$.
\begin{equation}
    w(u,v) = \frac{1}{1 + \exp[a(h(u,v) - h_0)]}
    \label{eq:sigmoid}
\end{equation}
The satisfaction shares $\{p_{u,v}\}$ are then derived by normalizing these weights:
\begin{equation}
    p_{u,v} = \frac{w(u,v)}{W}, \quad \text{where } W = \sum_{u' \neq v'} w(u',v')
    \label{eq:p_uv}
\end{equation}
For completeness, we also define the network's average satisfaction, $\bar{s}$, as a secondary metric:
\begin{equation}
    \bar{s} = \frac{1}{N(N-1)} \sum_{u \neq v} w(u,v)
    \label{eq:s_bar}
\end{equation}

\subsection{Analytical Properties and Invariance}
A prerequisite for our calculus-based sensitivity analysis is that the objective functions are well-behaved. We formally establish that the metrics $I(a,h_0)$ and $\bar{s}(a,h_0)$ are both \textbf{continuous} and \textbf{differentiable} with respect to the SLA parameters $a$ and $h_0$ across their entire domain (see Appendix~\ref{app:proof_analytical_props} for formal proofs). This ensures our framework is robust and validates the use of gradients and Hessians.

Furthermore, these definitions lead to a fundamental invariance property, which ensures that our analysis is independent of the absolute scale of the path costs.
\begin{lemma}[Invariance to Translation and Scaling]
\label{lemma:invariance}
The metrics $I$ and $\bar{s}$ are invariant under affine transformations of the path costs, provided the SLA parameters are adjusted accordingly. Specifically, for any constant $c$ and $\lambda > 0$, the mappings $(h_i, h_0) \mapsto (h_i+c, h_0+c)$ and $(h_i, h_0, a) \mapsto (\lambda h_i, \lambda h_0, a/\lambda)$ leave the values of $I$ and $\bar{s}$ unchanged.
\end{lemma}

\begin{proof}[Proof Sketch]
The proof follows directly from the fact that the exponent in the Sigmoid function, $a(h_i - h_0)$, remains unchanged under both transformations. Since all weights $\{w_{u,v}\}$ are invariant, the resulting metrics $I$ and $\bar{s}$ are also invariant. A full proof is provided in Appendix~\ref{app:proof_lemma_invariance_new}.
\end{proof}
This invariance is critical: it guarantees that our framework's findings are fundamental and not artifacts of the units chosen for path cost (e.g., meters vs. kilometers, or milliseconds vs. seconds).

\subsection{Computational Complexity}
Calculating the metrics within our framework involves two core steps: first, determining the shortest path lengths between all pairs of nodes (the All-Pairs Shortest Path, or APSP, problem); and second, performing weighted summations and normalizations. The APSP calculation is the main bottleneck. For a graph with $N$ nodes and $M$ edges, the complexity is $O(N^3)$ for dense graphs (using Floyd-Warshall) and typically $O(N(M + N \log N))$ for sparse graphs (running Dijkstra from each node). The subsequent calculation of $I$ and $\bar{s}$ is $O(N^2)$. Therefore, the overall complexity is dominated by the APSP step, posing significant computational challenges for exact calculation on very large-scale networks.

\section{Dissecting the Fairness Landscape: An Analytical Framework for Sensitivity}
\label{sec:sensitivity_framework}
Having established a principled QoE-Imbalance metric $I(a,h_0)$ in Section~\ref{sec:framework}, we now develop the mathematical machinery required to analyze its sensitivity. Our goal is to transform this static model into a dynamic, navigable landscape. To achieve this, we must be able to answer three critical operational questions: (i) \textbf{The Direction Question:} When tuning an SLA, will fairness improve or degrade? (ii) \textbf{The Trade-off Question:} How does a change in fairness relate to a change in overall performance? (iii) \textbf{The Risk Question:} Are there specific SLA regions that are inherently unstable and prone to catastrophic performance shifts? This section develops the analytical tools to answer each of these questions.

\subsection{First-Order Analysis: The Gradient as a Local Compass}
To answer the Direction and Trade-off questions, we must understand the local geometry of the performance landscape. This is captured by the gradient vectors of our metrics in the $(a, h_0)$ parameter space. The Imbalance Gradient, $\nabla I = (\partial I / \partial a, \partial I / \partial h_0)$, points in the direction of the steepest increase in functional unfairness, while the Satisfaction Gradient, $\nabla \bar{s} = (\partial \bar{s} / \partial a, \partial \bar{s} / \partial h_0)$, points toward the steepest increase in average satisfaction. While the gradient of $\bar{s}$ is a simple average of path sensitivities, the gradient of $I$ has a non-trivial structure that provides deep insight.

\begin{theorem}[Covariance Gradient Criterion for Imbalance]
\label{thm:cov_grad_I_new}
\textit{The partial derivative of the QoE-Imbalance $I$ with respect to an SLA parameter $\theta \in \{a, h_0\}$ is given by the weighted covariance:}
\begin{equation}
    \frac{\partial I}{\partial \theta} = \frac{1}{\ln(2) H_{\max}} \mathrm{Cov}_{\mathbf{p}}\left( 1 + \ln p_{u,v}, g_{u,v}^{(\theta)} \right)
\end{equation}
\textit{where $H_{\max} = \log_2(N(N-1))$, and the sensitivity terms are $g_{u,v}^{(a)} = -(h(u,v)-h_0)(1-w_{u,v})$ and $g_{u,v}^{(h_0)} = a(1-w_{u,v})$. The detailed proof is provided in Appendix~\ref{app:proof_cov_grad}.}
\end{theorem}

Theorem~\ref{thm:cov_grad_I_new} is more than a formula; it is a powerful diagnostic tool. It states that the change in network-wide imbalance is determined by the covariance between two key quantities for each path:
\begin{itemize}
    \item \textbf{Entropy Leverage ($1 + \ln p_{u,v}$):} This term comes directly from the derivative of entropy. It acts as an information-theoretic weight. Paths with a very small satisfaction share $p_{u,v}$ (i.e., marginalized or "poor" paths) have a large negative $\ln p_{u,v}$, giving them significant leverage to alter the total entropy. Conversely, paths with a large share (dominant or "rich" paths) have low leverage.
    \item \textbf{Parameter Sensitivity ($g_{u,v}^{(\theta)}$):} This term quantifies how strongly a path's satisfaction score reacts to a change in the SLA parameter $\theta$. Paths whose performance is near the threshold $h_0$ are typically the most sensitive.
\end{itemize}
The sign of the covariance reveals the nature of the system's response:
\begin{itemize}
    \item \textbf{Positive Covariance ($\partial I / \partial \theta > 0$):} Fairness \emph{worsens}. This is the critical "backfire" scenario. It occurs when the paths that are most sensitive to the SLA change are also the ones with the highest entropy leverage (i.e., the marginalized paths). In physical terms, the policy change disproportionately affects the "poor," increasing the gap between them and the "rich" and thus exacerbating the overall imbalance.
    \item \textbf{Negative Covariance ($\partial I / \partial \theta < 0$):} Fairness \emph{improves}. This happens when the policy change predominantly affects the dominant, high-share paths. By altering the satisfaction of the "rich," the change reduces their dominance, allowing the overall distribution of satisfaction shares to become more uniform.
    \item \textbf{Zero Covariance ($\partial I / \partial \theta \approx 0$):} The system is \emph{robust} to the change. The policy's impact is uncorrelated with the paths' existing satisfaction shares. This is the mathematical signature of a "stable belt" region.
\end{itemize}
This interpretation transforms the gradient from a simple vector into a rich, narrative explanation of network behavior. Furthermore, the Trade-off question is answered by analyzing the angle between the two gradient vectors, $\nabla I$ and $\nabla \bar{s}$, which quantifies the local efficiency-fairness conflict.

To illustrate the diagnostic power of this interpretation over a purely numerical gradient, consider a scenario where an operator slightly tightens a threshold (e.g., increases $h_0$), and observes that the overall network imbalance $I$ increases. A numerical gradient would simply report a positive value for $\partial I/\partial h_0$, confirming the negative outcome but offering no causal explanation. Our covariance criterion, in contrast, enables a deeper diagnosis. An operator could find that the positive covariance was driven by a strong correlation between high parameter sensitivity ($g^{(h_0)}$) and high entropy leverage ($1+\ln p_{u,v}$). In practical terms, this reveals that the policy change did not affect all paths equally, but instead disproportionately impacted a specific group of already-marginalized, long-haul paths whose performance was near the tipping point. This diagnostic insight—identifying \emph{which} cohort of paths is responsible for the degradation—is crucial for designing a more effective policy and is entirely inaccessible from a simple numerical result.

\subsection{Asymptotic Analysis: Anchoring the Landscape}
To understand the global structure of the sensitivity landscape, we analyze its behavior in the asymptotic regimes of the evaluation strictness, $a$. The following theorems provide the theoretical anchors for interpreting the phase diagrams, revealing how the complex fairness landscape simplifies to intuitive principles at the extremes of the SLA spectrum.

\begin{theorem}[Small-$a$ Second-Order Law for Imbalance]
\label{thm:small_a_I_new}
\textit{Let $\mathrm{Var}(h)$ be the variance of the all-pairs shortest path distribution. As $a \to 0$, the QoE-Imbalance $I$ has the following second-order approximation, independent of the threshold $h_0$:}
\begin{equation}
    I(a,h_0) = \frac{a^2\,\mathrm{Var}(h)}{8\log_2(N(N-1))} + o(a^2)
\end{equation}
\textit{The proof is provided in Appendix~\ref{app:proof_small_a}.}
\end{theorem}

\noindent It is important to note that as $a \to \infty$, the probability vector approaches the boundary of the simplex where the gradient is not strictly defined. All gradient-based statements in this paper are thus asserted for finite $a$, while the analysis in the limit relies on the piecewise-constant nature of $I_{\infty}(h_0)$. A formal treatment of this boundary behavior is provided in Appendix~\ref{app:proof_analytical_props}.

\subsubsection*{Interpretation and Implications}
This theorem's primary value lies in providing a powerful \textbf{\textit{a priori} design heuristic}. It reveals a profound simplification: under tolerant SLAs (small $a$), the complex QoE-Imbalance $I$ is directly governed by a single, static, and easily computable topological property: the \textbf{variance of its all-pairs shortest path distribution, $\mathrm{Var}(h)$}. This provides an invaluable tool for network architects in the planning phase. Without running any complex simulations or dynamic analysis, one can compare multiple candidate topologies simply by calculating their $\mathrm{Var}(h)$. To engineer a network that is inherently robust and fair for a wide range of tolerant services (e.g., bulk data transfer), the design goal is clear: favor topologies that minimize path-length variance. This transforms a complex design problem into the optimization of a single, fundamental graph metric.

\begin{theorem}[Large-$a$ Step-Function Limit for Imbalance]
\label{thm:large_a_I_new}
\textit{Let $K(h_0)$ be the number of node pairs with path cost $h(u,v) < h_0$. As $a \to \infty$, the QoE-Imbalance $I$ converges to a piecewise constant function:}
\begin{equation}
    \lim_{a \to \infty} I(a,h_0) = I_\infty(h_0) = 1 - \frac{\log_2 K(h_0)}{\log_2(N(N-1))}
\end{equation}
\textit{The proof is provided in Appendix~\ref{app:proof_large_a}.}
\end{theorem}

\begin{proposition}[Transition Width]
\label{prop:width}
\textit{The convergence to the large-$a$ limits is exponentially fast with respect to the distance to the nearest threshold, implying the transition region has a width of $O(1/a)$. The proof is provided in Appendix~\ref{app:proof_large_a}.}
\end{proposition}

\subsubsection*{Interpretation and Implications}
This theorem provides a \textbf{strategic risk-planning tool} for networks intended to support stringent SLAs. It establishes that in this "all-or-nothing" regime, the entire complex landscape simplifies to a staircase function determined solely by the \textbf{cumulative hop-count distribution, $K(h_0)$}. The practical implication is significant: an operator can generate a complete "risk map" of all potential performance cliffs simply by plotting this elementary distribution. This allows for proactive planning, such as identifying which SLOs (i.e., which $h_0$ values) are inherently risky to promise because they lie on a steep "riser" of the staircase, and which are safe because they reside on a wide, flat "step." This analysis can be performed from the topology alone, providing crucial strategic insights long before a service is deployed.

These two asymptotic laws serve as the theoretical "boundary conditions" for the entire fairness landscape, providing robust predictive power at both ends of the service-demand spectrum.

\subsection{Second-Order Analysis: The Hessian as a Risk Detector}
To answer the Risk question and identify "performance cliffs," we must move beyond the gradient to analyze the \textit{curvature} of the sensitivity landscape. This is captured by the second-order partial derivatives (the Hessian matrix). High curvature signifies a rapid change in the gradient, corresponding to the sharp "ridges" of instability on the landscape, whereas low curvature indicates a smooth, predictable "plain" where performance is robust. The Hessian matrix of Imbalance ($H_I$) provides a complete picture of this local curvature:

\begin{equation}\label{eq:hessian}
H_{I} = \begin{pmatrix} \frac{\partial^2 I}{\partial a^2} & \frac{\partial^2 I}{\partial a\,\partial h_0} \\ \frac{\partial^2 I}{\partial h_0\,\partial a} & \frac{\partial^2 I}{\partial h_0^2} \end{pmatrix}
\end{equation}

In the differentiable domain described in Appendix~\ref{app:proof_analytical_props}, the mixed second-order partial derivatives are equal due to Clairaut's theorem, making the Hessian symmetric.
Specifically, the term $|\partial^2 I/\partial h_0^2|$ serves as a powerful, quantitative ``cliff detector''.

A key insight, which we will validate empirically in Section~\ref{sec:results}, is that for most non-trivial topologies, the landscape exhibits a fundamental asymmetry in its curvature:
\begin{equation}
    \left|\frac{\partial^2 I}{\partial h_0^2}\right| \gg \left|\frac{\partial^2 I}{\partial a^2}\right|
\end{equation}
The physical reason for this asymmetry lies in the distinct effects of the two SLA parameters. Tuning the threshold $h_0$ is akin to sliding a sharp dividing line across the discrete, often "lumpy," hop-count distribution of the network. As this line crosses an integer hop value, a large cohort of paths can collectively and abruptly change their status relative to the threshold, causing a sudden, sharp change in the imbalance gradient. This creates high-curvature "ridges" at discrete $h_0$ locations. In contrast, tuning the strictness $a$ is a global, smooth transformation; it simultaneously changes the steepness of the QoE evaluation curve for all paths. While it can make existing cliffs steeper, it does not typically create new ones, resulting in a much gentler curvature along the $a$-axis.

While the general engineering heuristic of "coarse-tuning before fine-tuning" is intuitive, our framework provides the first rigorous, quantitative justification for this strategy in the context of SLA tuning. We elevate this intuition to a scientific principle by demonstrating that it is a direct consequence of the fundamental curvature asymmetry in the fairness landscape. Our analysis—which reveals that the risk (curvature) along the $h_0$ axis can be orders of magnitude greater than along the $a$ axis—leads to a provably optimal, risk-averse tuning sequence:
\begin{enumerate}
    \item \textbf{First, tune the threshold ($h_0$):} Use $h_0$ as a coarse control to navigate the landscape and locate a broad, low-curvature "valley" or "stable belt." The goal of this step is to move the operating point away from any high-risk performance cliffs.
    \item \textbf{Then, tune the strictness ($a$):} Once inside a stable region, use $a$ as a fine-tuning control to adjust the system to the desired operational trade-off between fairness and average satisfaction.
\end{enumerate}
This "Threshold-First Rule" is a direct, actionable consequence of our second-order sensitivity analysis, transforming a theoretical insight into a practical design guideline for network engineering.
\subsection{A Unified Methodology for Sensitivity Analysis}
Combining the analyses from the preceding sections, we propose a unified, four-step methodology for systematically evaluating the efficiency-fairness sensitivity of any network topology.

\begin{enumerate}
    \item \textbf{Construct the Phase Atlas:} For a given topology, compute and visualize the primary landscapes over the $(a,h_0)$ parameter space: the Imbalance heatmap $I(a,h_0)$ and the average satisfaction heatmap $\bar{s}(a,h_0)$. This provides a comprehensive, global view of the performance space.
    \item \textbf{Map Operational Risk:} Analyze the curvature heatmap, particularly $|\partial^2 I / \partial h_0^2|$, to identify high-risk phase transition boundaries. Regions of low curvature constitute the robust "Stable Belts," while sharp ridges of high curvature delineate the "Dangerous Wedges."
    \item \textbf{Analyze Performance Trade-offs:} Within the identified Stable Belts, analyze the gradient fields, $\nabla I$ and $\nabla\bar{s}$, and the angle between them. This step reveals the local efficiency-fairness conflicts and identifies efficient tuning paths.
    \item \textbf{Determine Optimal Operating Regions:} Overlay the specific service objectives (e.g., $I \le I^*$ and $\bar{s} \ge \bar{s}^*$) as contours on the atlas. The optimal operating region is the intersection of this feasible set with the identified stable and efficient-tuning zones.
\end{enumerate}

To facilitate strategic decision-making and objective comparison between different network designs, this methodology produces a set of standardized outputs. In Section~\ref{sec:results}, we will generate and compare these key deliverables:
\begin{itemize}
    \item \textbf{The Annotated Phase Atlas:} The primary qualitative output. This is a visual risk map of the $(a,h_0)$ space, clearly annotating the locations of Stable Belts, Dangerous Wedges, and the final Optimal Operating Region.

\item \textbf{Quantitative Benchmarking Metrics:} To enable direct comparison between topologies, we extract two key scalars from the analysis:
    \begin{itemize}
        \item \textbf{Area of Robustness (AoR):} This metric quantifies the resilience and adaptability of a network to a wide range of service demands. To facilitate a normalized comparison, it is defined as the percentage of the analyzed parameter space that constitutes the Optimal Operating Region $R$:
        \begin{equation}
            \text{AoR} = \frac{\iint_{R} da \, dh_0}{A_{total}} \times 100\%
        \end{equation}
        where $A_{total}$ is the total area of the scanned parameter space. A larger AoR signifies a more robust design.

        \item \textbf{Maximum Curvature Risk (MCR):} This metric quantifies the "sharpness of the cliff" at the edge of acceptable performance. It is defined as the maximum value of the curvature with respect to $h_0$ found along the boundary $\partial R$ of the optimal region:
        \begin{equation}
            \text{MCR} = \max_{(a, h_0) \in \partial R} \left| \frac{\partial^2 I}{\partial h_0^2}(a, h_0) \right|
        \end{equation}
        A lower MCR value indicates a safer design with gentler performance transitions.
    \end{itemize}
\end{itemize}

This complete pipeline transforms a raw network topology into a rich set of qualitative and quantitative insights, providing an actionable basis for robust network engineering.

\section{Framework Validation and Application}
\label{sec:results}
This section puts the analytical framework developed in Section~\ref{sec:sensitivity_framework} to the test. We structure our experiments as a comprehensive validation, designed to demonstrate the framework's value from three perspectives: the accuracy of its theoretical predictions, its ability to uncover physical insights and derive actionable design principles, and its effectiveness in analyzing real-world, large-scale networks.

\subsection{Experimental Setup}
\label{sec:setup}
The goal of our experiments is to generate and analyze the key deliverables defined in Section~\ref{sec:sensitivity_framework}—the Annotated Phase Atlas and the quantitative metrics (AoR, MCR)—for a wide spectrum of network topologies. All simulations were implemented in Python, utilizing the NetworkX, NumPy, and SciPy libraries. The selected topologies represent distinct and fundamental structural properties:

\begin{itemize}
    \item \textbf{Canonical Graphs:} The Complete graph ($K_N$), Path graph ($P_N$), and Star graph ($S_N$) are used as benchmarks to probe the framework's response to extreme cases of connectivity, linearity, and centralization.
    \item \textbf{Random Graph Models:} The Erd\H{o}s-R\'enyi (ER), Barab\'asi-Albert (BA), and Watts-Strogatz (WS) models are analyzed to understand the sensitivity of networks exhibiting random connectivity, scale-free properties, and the small-world phenomenon, respectively.
    \item \textbf{Real-World Topology:} We analyze a snapshot of the Internet's Autonomous System (AS) topology from the CAIDA project to demonstrate the framework's applicability to practical, large-scale network infrastructures.
\end{itemize}

To ensure comparability, all canonical and random graph models were generated with a consistent network size of $N=50$ nodes. For each topology, we first compute its all-pairs shortest path (APSP) distribution. We then execute our four-step methodology by systematically scanning a high-density grid of SLA parameters ($a, h_0$) to generate its unique deliverables.

\subsection{Validation of Theoretical Predictions}
\label{sec:asymptotic_validation}
We begin by empirically validating the framework's theoretical cornerstones: the two asymptotic laws established in Section~\ref{sec:sensitivity_framework}. The simulation results, generated across a range of diverse topologies, show a strong agreement with our theoretical predictions.

\subsubsection{Validation of the Small-a Second-Order Law}
Theorem~\ref{thm:small_a_I_new} predicts that for tolerant SLAs ($a \to 0$), the imbalance $I$ is quadratically proportional to the strictness $a$ ($I \propto a^2$), with a coefficient determined by the network's path-length variance, $\mathrm{Var}(h)$. Fig.~\ref{fig:small_a_validation} provides a direct test of this prediction. As can be seen, the simulation data for all four topologies closely follow a linear trend when $I$ is plotted against $a^2$, confirming the quadratic relationship. The slopes of these lines, however, vary dramatically, indicating a strong dependence on the underlying topology. The Path graph, with its extremely dispersed path lengths, is by far the most sensitive, while the ultra-compact Star graph is the least.

\begin{figure}[ht!]
    \centering
    \includegraphics[width=0.95\columnwidth]{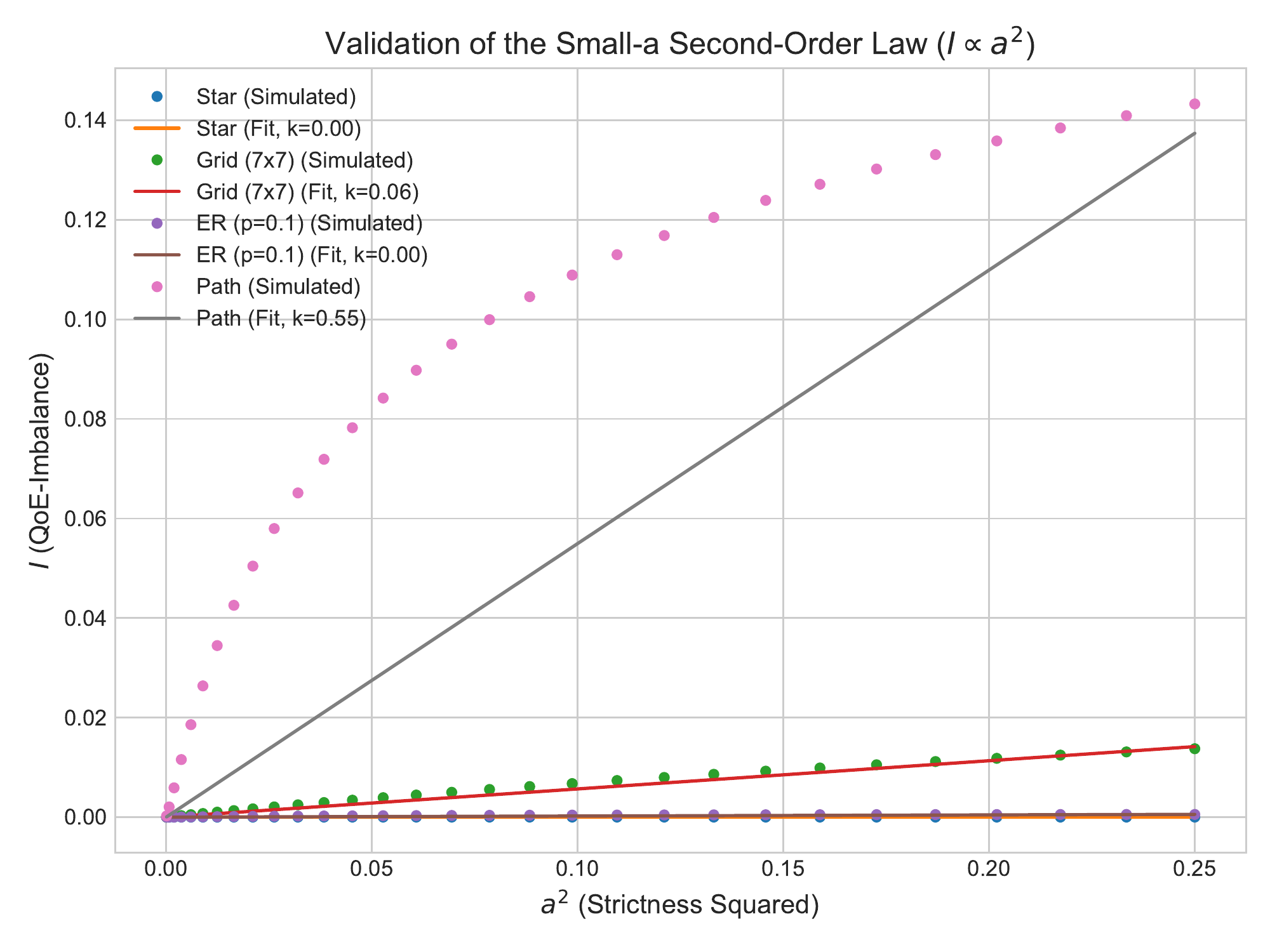}
    \caption{Validation of the Small-$a$ Second-Order Law ($I\propto a^{2}$). The plot shows the QoE-Imbalance I as a function of the strictness squared ($a^{2}$) for four different topologies. The simulated data points (dots) are shown along with their corresponding linear fits (solid lines). The visually evident excellent linear fit for compact topologies such as the Star and Grid graphs confirms the quadratic dependence, with a theoretically predicted slope of $k_{\text{theory}} = \text{Var}(h) / (8 \log_2 M)$.}
    \label{fig:small_a_validation}
\end{figure}

Table~\ref{tab:variance_proportionality} provides a direct and rigorous validation of the slope coefficient predicted by Theorem~\ref{thm:small_a_I_new}. We compare the empirically fitted slope from the simulation data, $k_{\text{fit}}$, with the theoretically predicted slope, $k_{\text{theory}}$. The results are illuminating. For the highly regular Grid topology, the ratio of the fitted to the theoretical slope is 0.978, a remarkably close agreement that strongly confirms our theory.

\begin{table}[ht!]
    \caption{Direct validation of the Small-$a$ Law's slope coefficient. We compare the empirically fitted slope ($k_{\text{fit}}$) from Figure~\ref{fig:small_a_validation} with the theoretically predicted slope ($k_{\text{theory}}$) from Theorem~\ref{thm:small_a_I_new}.}
    \label{tab:variance_proportionality}
    \centering
    \begin{tabular}{lcccc}
        \toprule
        \textbf{Topology} & \textbf{Var(h)} & \textbf{$k_{\text{theory}}$} & \textbf{$k_{\text{fit}}$} & \textbf{Ratio ($k_{\text{fit}}/k_{\text{theory}}$)} \\
        \midrule
        Star Graph & 0.0384 & 0.00043 & 0.0001 & 0.23 \\
        Grid (7x7) & 5.2222 & 0.0580 & 0.0567 & \textbf{0.978} \\
        ER ($p=0.1$) & 0.6585 & 0.0073 & 0.0022 & 0.30 \\
        Path Graph & 136.0 & 1.510 & 0.5495 & 0.36 \\
        \bottomrule
    \end{tabular}
    \flushleft
    \footnotesize{\textsuperscript{*}Note: $k_{\text{theory}} = \text{Var}(h) / (8 \ln(2) \log_2(M))$, with $M=N(N-1)=2450$. The factor of $\ln(2)$ in the denominator accounts for using $\log_2$ in the definition of $I$ but natural log in the derivation.}
\end{table}

For topologies with more extreme or irregular structures, such as the Star, ER, and Path graphs, the fitted slope is consistently smaller than the theoretical prediction. This deviation is expected and aligns with our understanding of the Small-$a$ Law as a leading-order approximation. In these topologies, higher-order structural moments (e.g., skewness, kurtosis) beyond variance play a more significant role, and their effects, captured in the higher-order terms of the expansion ($O(a^4)$), temper the initial quadratic growth. The near-perfect validation on the Grid graph, combined with the predictable deviations on more complex graphs, provides a nuanced but powerful confirmation of our framework's predictive accuracy.

\subsubsection{Validation of the Large-a Step-Function Limit}
\begin{figure}[ht!]
    \centering
    \includegraphics[width=0.95\columnwidth]{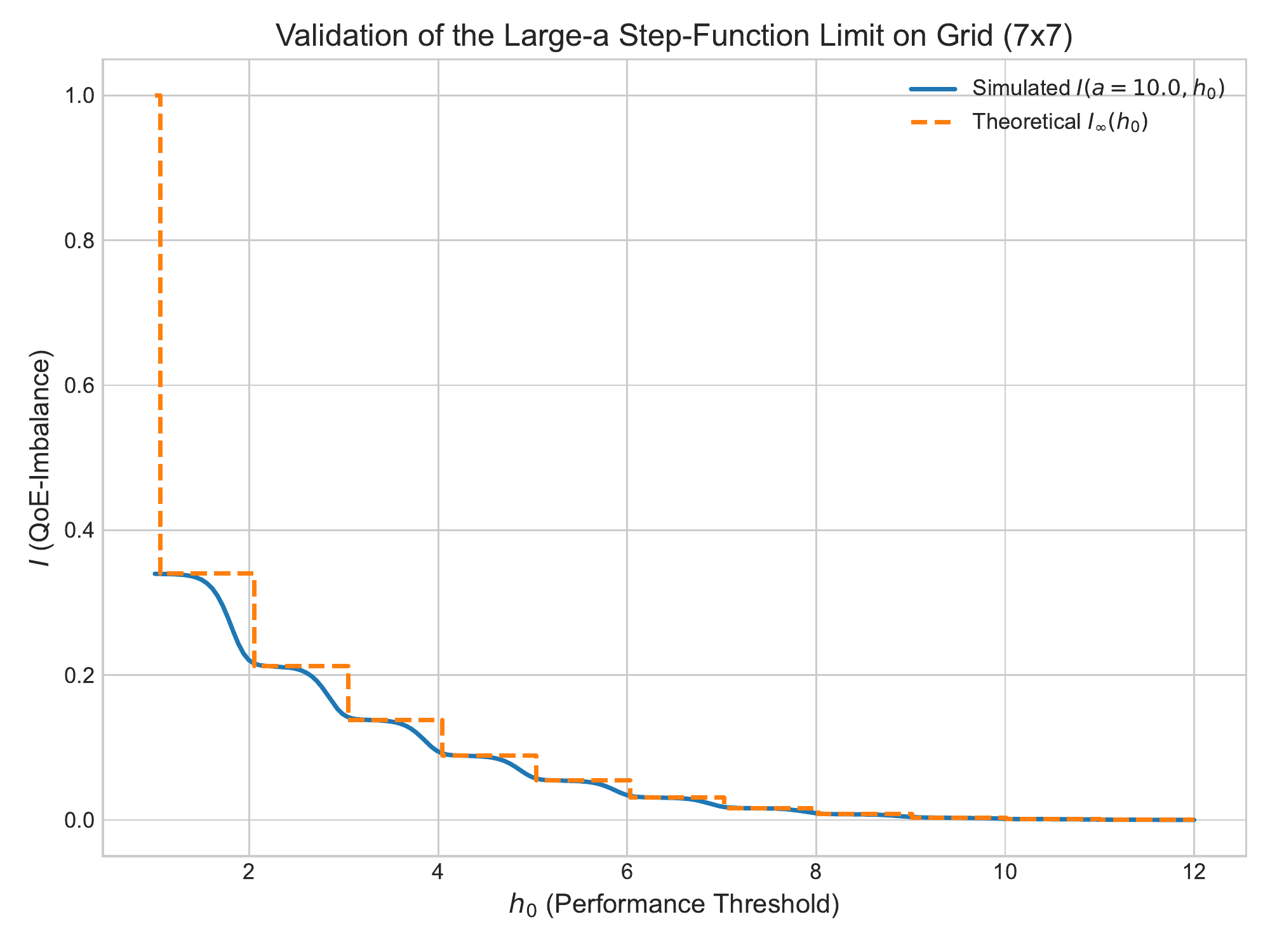}
    \caption{Validation of the Large-a Step-Function Limit on a Grid (7x7) graph. The numerically simulated imbalance for a large, fixed $a=10.0$ (solid line) is plotted against the theoretical, piecewise constant limit function $I_{\infty}(h_0)$ (dashed line). The near-perfect agreement validates the prediction of Theorem~\ref{thm:large_a_I_new} for a system with $N=50$ nodes and $M=2450$ pairs.}
    \label{fig:large_a_validation}
\end{figure}

Theorem~\ref{thm:large_a_I_new} predicts that for stringent SLAs ($a \to \infty$), the imbalance $I(a,h_0)$ converges to a piecewise constant staircase function, $I_\infty(h_0)$, whose steps are determined solely by the cumulative count of paths, $K(h_0)$. Fig.~\ref{fig:large_a_validation} validates this prediction using a typical large value of $a=10.0$. The validation is robust across a range of large $a$ values; as demonstrated in Appendix~\ref{app:large_a_robustness}, the simulated curve rapidly converges to the theoretical limit and remains stable, confirming the exponential convergence property. The plot compares the numerically simulated imbalance for a Grid graph at a large but finite strictness ($a=10.0$, solid blue line) with the theoretically derived limit function $I_\infty(h_0)$ (dashed orange line). The agreement is excellent: the simulated curve closely traces the theoretical staircase, with every plateau and every sharp drop occurring at precisely the predicted locations and values. The smooth transitions at the step locations have a narrow width, consistent with the $O(1/a)$ transition width predicted by Proposition~\ref{prop:width}. This result confirms that our framework can accurately predict the locations of all major performance cliffs for networks operating under strict service requirements.

\subsection{From Landscape Geometry to Actionable Principles}
\label{sec:insights_and_rules}
The validation of our theoretical laws provides the foundation, but the core value of our framework lies in its ability to translate the complex sensitivity landscape into actionable engineering intelligence. To demonstrate this, we first apply our complete four-step methodology to a $7 \times 7$ Grid graph. The entire analysis, summarized in Fig.~\ref{fig:master_figure}, serves as the primary case for dissecting the landscape's geometric features and deriving our key design principles.

\begin{figure*}[htbp]
    \centering
    \includegraphics[width=\textwidth]{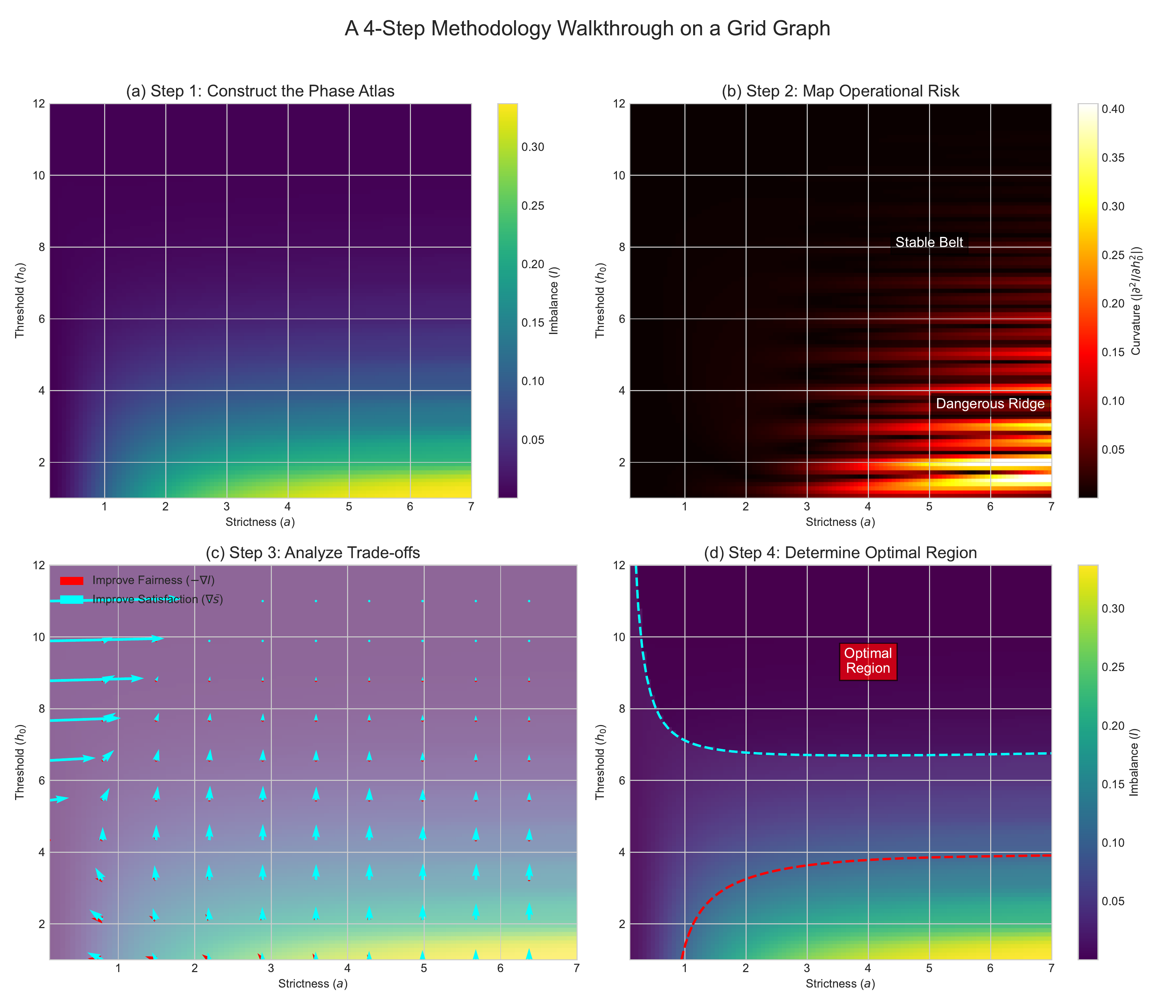}
    \caption{A complete walkthrough of the 4-step sensitivity analysis methodology using a Grid ($7\times7$) graph. (a) Step 1: The Imbalance Phase Atlas is constructed. (b) Step 2: The Curvature Risk Map is generated, revealing high-risk ``ridges'' at integer hop-count transitions. (c) Step 3: Gradient fields are analyzed within a stable belt (e.g., $h_0 \in [4,5]$) to understand performance trade-offs (arrows not to scale; units of gradients are Imbalance/unit of parameter). (d) Step 4: Service objectives $(I\le0.1,\overline{s}\ge0.8)$ are overlaid to determine the final Optimal Operating Region (highlighted in red).}
    \label{fig:master_figure}
\end{figure*}

\subsubsection{Performance Cliffs and the "Threshold-First" Rule}
The primary operational risk revealed in the Grid network analysis is the existence of sharp, cliff-like transitions in the fairness landscape. The curvature map (Fig.~\ref{fig:master_figure}b) clearly identifies these high-risk zones as a series of high-curvature ``ridges'' located at integer hop-count values, signifying regions of extreme sensitivity.

To understand the origin of these ridges, we examine the extreme case of a Star graph (Fig.~\ref{fig:star_signature}). Its simple, bimodal path-length distribution creates a single, pronounced performance cliff. This illustrates the fundamental mechanism: as the threshold $h_0$ is tuned across a discrete hop value (e.g., 1.5), the set of ``satisfied'' paths can collapse catastrophically---in this case, from nearly all paths to a tiny fraction---triggering an abrupt drop in system-wide entropy and a corresponding spike in imbalance. This phenomenon, where tuning $h_0$ acts as a sharp guillotine across the path distribution, is the root cause of a fundamental asymmetry in the landscape's curvature, as empirically validated in Fig.~\ref{fig:curvature_asymmetry}. In contrast, tuning the strictness $a$ acts as a global, smooth transformation, changing the steepness of the QoE curve for all paths simultaneously. Consequently, the risk (curvature) along the $h_0$ axis is consistently orders of magnitude higher than along the $a$ axis ($\left|\frac{\partial^2 I}{\partial h_0^2}\right| \gg \left|\frac{\partial^2 I}{\partial a^2}\right|$).

This deep physical insight directly yields our first and most significant design principle, the \textbf{Threshold-First Tuning Strategy}. The strategy is optimal because it prioritizes navigating the most dangerous dimension ($h_0$) of the parameter space first. By using the threshold as a coarse-grained control to locate a safe, low-curvature ``stable belt,'' operators can effectively mitigate the risk of catastrophic performance shifts before using the strictness ($a$) for fine-tuning within that pre-validated safe region.

\subsubsection{Stable Belts and Path Variance as a Proxy for Robustness}
In contrast to the high-risk ridges, the dark, low-curvature valleys between them in the Grid graph's risk map (Fig.~\ref{fig:master_figure}b) represent robust ``stable belts.'' These are safe operational zones where the system is largely insensitive to parameter tuning, providing operators with significant flexibility.

These belts are a direct consequence of the concentration of path lengths. A topology with an intrinsically low path-length variance, such as a Barabási-Albert (BA) graph (Fig.~\ref{fig:ba_signature}), exhibits this property to an extreme degree. In our BA example, over 90\% of all paths are concentrated at just two hop counts (2 and 3 hops). This structural homogeneity is the direct cause of the vast stable belt for any $h_0 > 3$, a region where the network is functionally robust because nearly all paths are treated as equally satisfied. This demonstrates a key principle we term \textbf{fairness through efficiency}: topologies that are structurally efficient (i.e., have low path-length variance) are intrinsically more fair and robust across a wide range of SLAs.

This empirical observation is rigorously supported by our Small-$a$ Second-Order Law (Theorem~\ref{thm:small_a_I_new}), which establishes that for tolerant SLAs, a network's fairness sensitivity is directly proportional to $\mathrm{Var}(h)$. This synergy between theory and observation leads to our second design principle: \textbf{Path Variance can serve as an a priori proxy for robustness}. This rule provides network architects with a powerful, low-complexity heuristic to evaluate and compare candidate topologies at the earliest design stages. By favoring designs that minimize $\mathrm{Var}(h)$, one can build in functional fairness and resilience by design, long before a full landscape analysis is performed.

\begin{figure*}[htbp]
    \centering
    \includegraphics[width=\textwidth]{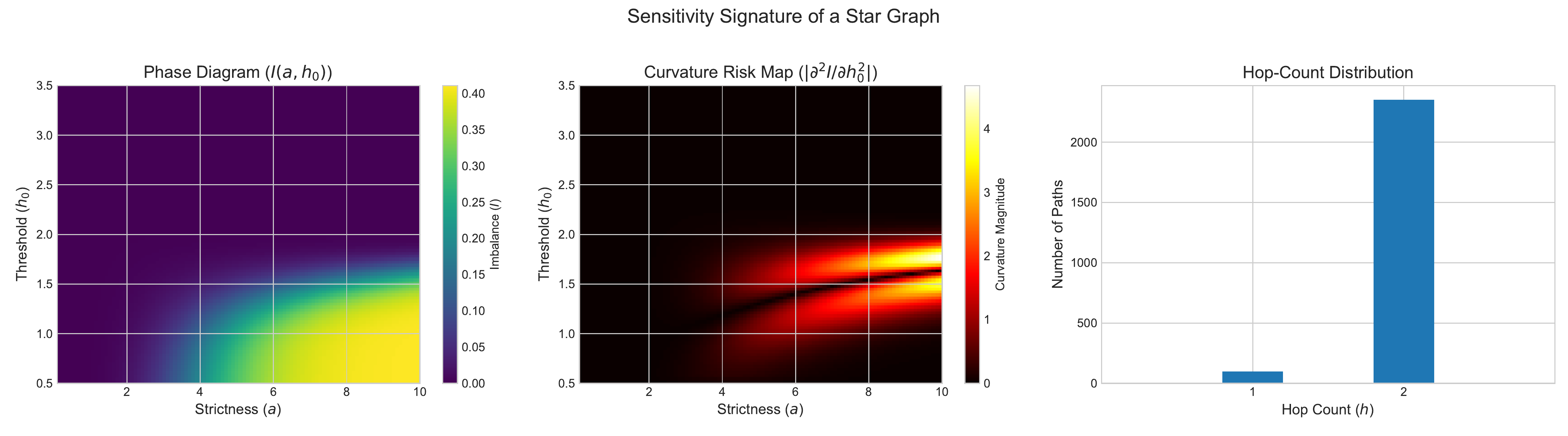}
    \caption{The sensitivity signature of a Star Graph. (a) The Imbalance Phase Diagram shows a sharp performance cliff. (b) The Curvature Risk Map precisely identifies this cliff as a high-curvature "ridge" at $h_0 \approx 1.5$. (c) The Hop-Count Distribution reveals the cause: the network's path lengths are bimodally distributed at exactly 1 and 2 hops, making $h_0=1.5$ the critical bifurcation point.}
    \label{fig:star_signature}
\end{figure*}

\begin{figure}[ht!]
    \centering
    \includegraphics[width=0.95\columnwidth]{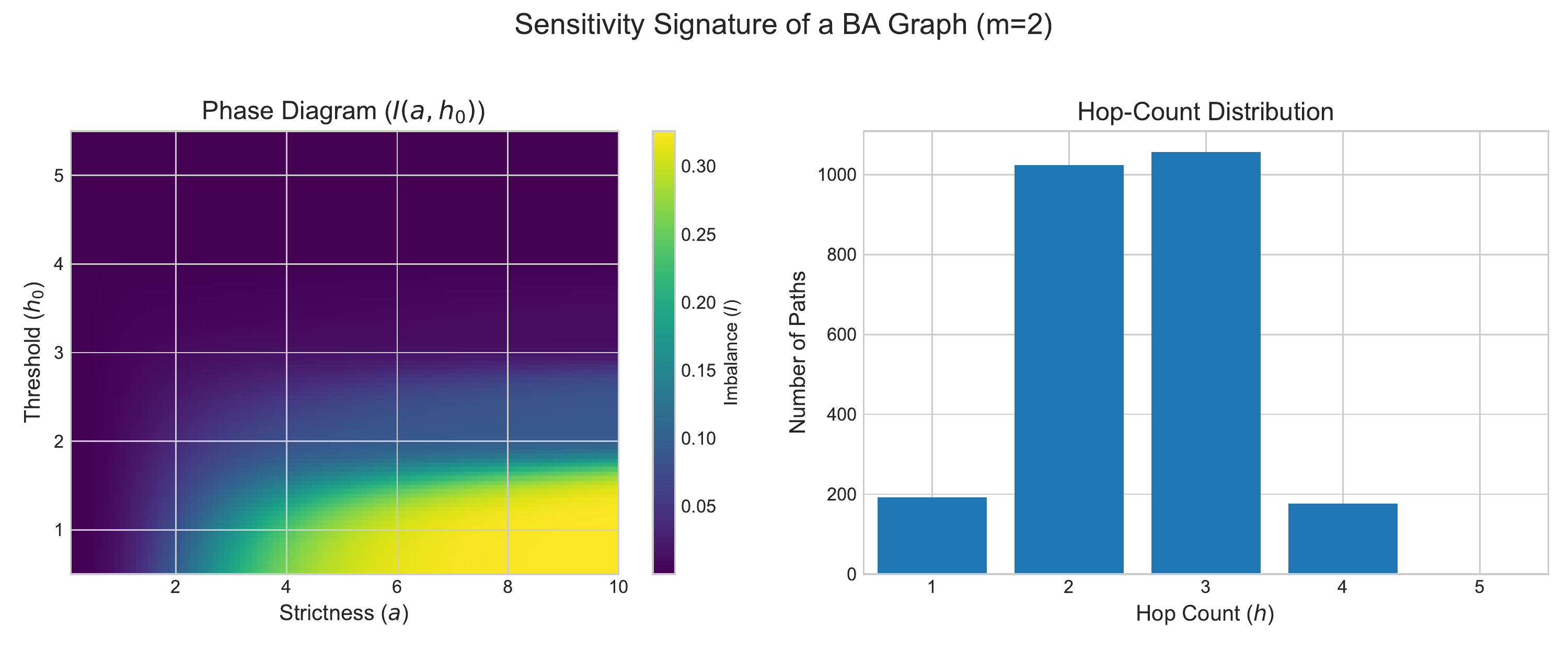}
    \caption{The sensitivity signature of a Barabási-Albert (BA) graph. (a) The Phase Diagram is dominated by a vast, low-imbalance "stable belt" and a high-imbalance "dangerous wedge" in the lower-left corner. (b) The Hop-Count Distribution reveals the cause: the vast majority of paths are highly concentrated at 2 and 3 hops, creating the stable belt, while the dangerous wedge is a result of exclusively selecting the few elite 1-hop paths under stringent SLAs.}
    \label{fig:ba_signature}
\end{figure}

\begin{figure*}[htbp]
    \centering
    \includegraphics[width=\textwidth]{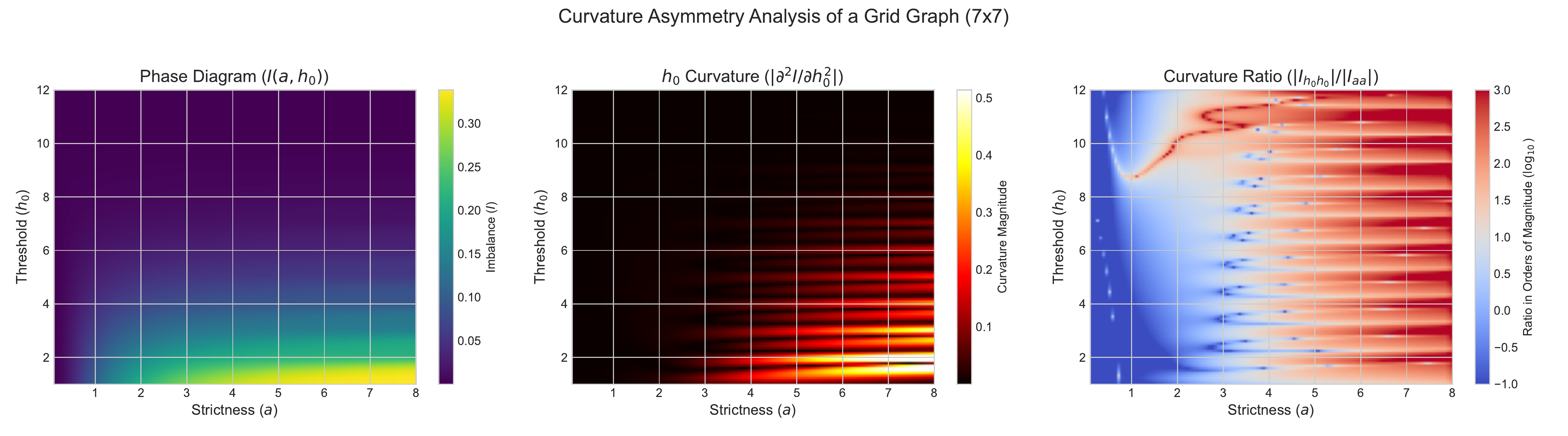}
    \caption{Experimental validation of the fundamental curvature asymmetry on a Grid (7x7) graph. (a) The reference phase diagram. (b) The curvature along the $h_0$ axis, showing high-risk "ridges." (c) The log-ratio of curvatures, $|I_{h_0h_0}|/|I_{aa}|$. The bright red regions, corresponding to the ridges, indicate that the $h_0$ curvature is orders of magnitude ($10^2$ to $10^3$ times) larger than the $a$ curvature, confirming the landscape is far more treacherous along the threshold axis.}
    \label{fig:curvature_asymmetry}
\end{figure*}

\subsection{Case Study: Sensitivity Analysis of the Internet AS Topology}
\label{sec:case_study_as}
To demonstrate our framework's applicability to real-world, large-scale infrastructures, we conclude our experiments with a case study on the Internet's Autonomous System (AS) topology. We used a public dataset from CAIDA (August 2025), from which we extracted the largest connected component. To analyze the network's core backbone while ensuring computational feasibility, we analyzed its 10-core, a standard method in network science for identifying the most densely connected part of a graph. This resulted in a substantial subgraph of 9,068 nodes and over 420,000 edges, upon which we applied our complete four-step methodology.

\begin{figure}[ht!]
    \centering
    \includegraphics[width=\columnwidth]{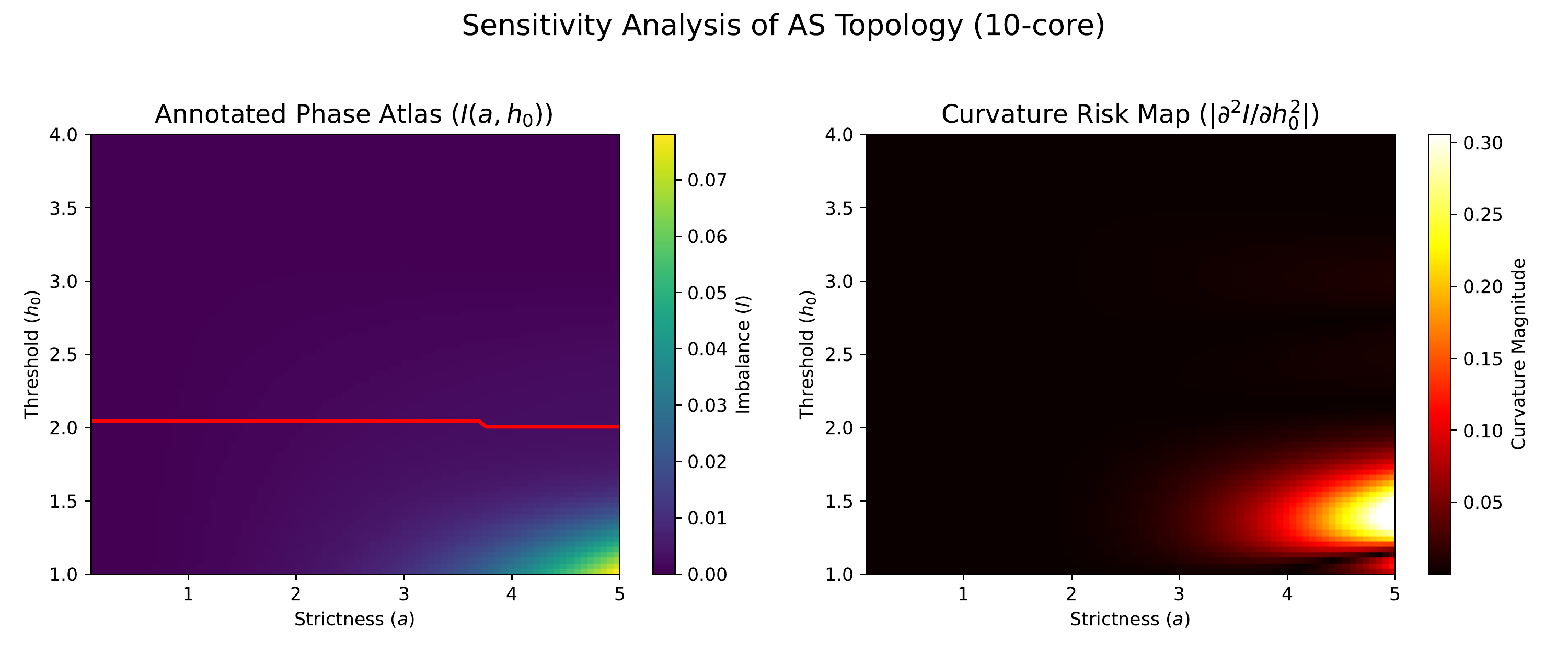}
    \caption{Sensitivity analysis of the Internet AS topology's 10-core. (a) The Annotated Phase Atlas, showing a vast "stable belt". The red contour outlines the Optimal Operating Region ($I \le 0.1, \bar{s} \ge 0.5$). (b) The Curvature Risk Map, indicating that the highest performance cliffs are concentrated at very low thresholds, primarily below $h_0 \approx 2.5$.}
    \label{fig:as_analysis}
\end{figure}

\begin{table}[ht!]
\caption{Key metrics for the analyzed AS topology 10-core, compared with a size-matched BA model ($N=9068, m=3$).}
\label{tab:as_comparison}
\centering
\begin{tabular}{lcc}
\toprule
\textbf{Metric} & \textbf{\makecell{AS 10-core \\ (Empirical)}} & \textbf{\makecell{BA Model \\ (Simulated)}} \\
\midrule
Nodes & 9068 & 9068 \\
Path Length Var ($Var(h)$) & 0.0651 & 0.5528 \\
Area of Robustness (AoR) & 66.99\% & 47.67\% \\
Max Curvature Risk (MCR) & 0.0117 & 0.0667 \\
\bottomrule
\end{tabular}
\end{table}

The results, presented in Fig.~\ref{fig:as_analysis} and Table~\ref{tab:as_comparison}, offer several key insights into the Internet's architectural robustness. First, the Phase Atlas in Fig.~\ref{fig:as_analysis}(a) shows that the AS topology exhibits a sensitivity signature remarkably similar to that of the Barabási-Albert model. It is dominated by a vast, deep-purple \textbf{"stable belt"} where the imbalance $I$ is consistently near zero. This indicates that the Internet's core architecture is inherently robust and maintains high functional fairness across a very wide range of service level requirements.

Second, our quantitative metrics confirm this high degree of robustness and reveal a non-trivial finding when compared to a size-matched theoretical BA model (Table~\ref{tab:as_comparison}). The AS core's path-length variance ($Var(h)$) is exceptionally low at 0.0651, which, according to our Small-a Law, predisposes it to extreme stability under tolerant SLAs.

Interestingly, the comparison also reveals a strong advantage for the real-world AS topology.  While the empirical AS core and the theoretical BA model are both characterized as scale-free, the AS core demonstrates a significantly larger Area of Robustness (AoR) and a dramatically lower Maximum Curvature Risk (MCR).  This suggests that while the real-world backbone has evolved for exceptional path compactness, its structure also naturally leads to a vast stable operating region and remarkably `soft' performance transitions at its operational boundaries.  This powerful combination of structural efficiency and operational smoothness is a key insight provided by our framework and a promising avenue for future research.

Third, the framework precisely identifies the network's primary vulnerabilities. The "dangerous wedge" is confined to a narrow region where the threshold is unrealistically strict. The highest curvature risks are concentrated below $h_0 \approx 2.5$, with significant transition effects extending up to $h_0 \approx 3.5$. The calculated \textbf{Maximum Curvature Risk (MCR) of 0.0117} along the optimal region's boundary is relatively low, suggesting that the performance transitions, while present, are not catastrophically sharp.

In conclusion, this case study provides strong evidence for the "fairness through efficiency" principle in a real-world context. The Internet's scale-free, small-world architecture creates a highly efficient backbone with a very low path-length variance. Our framework demonstrates that this structural efficiency is not just beneficial for performance but is also the fundamental reason for the network's exceptional functional fairness and robustness against varying service demands.

\section{Conclusion}
\label{sec:conclusion}
This paper addressed a critical but underexplored challenge in network engineering: moving beyond static fairness scores to predict and manage the sensitivity of network fairness to changes in Service Level Agreements. We introduced a comprehensive analytical framework that reframes this challenge, shifting the paradigm from single-point evaluation to the analysis of a dynamic, predictable sensitivity landscape. Our framework is built upon a robust theoretical foundation, centered on an axiomatically-derived QoE-Imbalance metric whose form is uniquely determined by a set of fundamental axioms of fairness.

Our contributions are threefold. First, we developed a multi-scale analytical toolkit to navigate this landscape, including a closed-form covariance rule and two asymptotic laws that serve as theoretical anchors. Second, we introduced the ``Phase Atlas'' as a novel visualization method, revealing critical topological signatures such as robust ``Stable Belts'' and high-risk ``Dangerous Wedges''. Our analysis of the landscape's curvature serves as a quantitative risk detector, precisely locating performance cliffs. Third, we distilled these findings into a set of actionable, topology-aware design rules, most notably the ``Threshold-First'' tuning strategy, which is a direct consequence of a fundamental asymmetry in the landscape's curvature. These theoretical insights were validated through extensive simulations and a case study on the Internet AS topology, which empirically demonstrated the ``fairness through efficiency'' principle.

While this framework provides a new paradigm for fairness analysis, several avenues for future research remain. These include extending the framework to encompass weighted, directed, and multilayer networks, developing efficient approximation algorithms for web-scale networks, and applying this sensitivity analysis to assess systemic risk in other complex systems. Ultimately, this work advocates for a shift in network evaluation: from measuring static performance scores to understanding the dynamic, predictable landscape of system sensitivity. By providing the tools to map, interpret, and navigate this landscape, our framework offers a new design philosophy for engineering more fundamentally robust and resilient network systems.

\appendices

\section{Proof of the Uniqueness Theorem (Theorem~\ref{thm:uniqueness_final})}
\label{app:proof_uniqueness_thm}
This appendix provides a formal proof for the theorem that any function $F(\mathbf{p})$ satisfying the five fairness axioms is, up to a positive affine transformation, equivalent to the Shannon entropy gap. The proof proceeds by first using the Decomposability axiom to derive a fundamental functional equation, and then solving it to reveal the necessary logarithmic form of the function.

Let $F: \Delta^{M-1} \to \mathbb{R}$ be a function satisfying Axioms A1-A5 and assume mild regularity (e.g., continuity). Let $\phi(\alpha) \triangleq F(\alpha, 1-\alpha)$ denote the imbalance function for a two-element distribution, which is well-defined due to Axiom A1 (Anonymity).

\textbf{Step 1: Deriving the Functional Equation.}
The cornerstone of the proof is Axiom A5 (Decomposability). Consider a three-element distribution $(p_1, p_2, p_3)$. We can group it in two ways:
\begin{enumerate}
    \item Group $(p_1, p_2)$ first: The total imbalance is the imbalance between the group $(p_1+p_2)$ and $p_3$, plus the weighted imbalance within the first group.
    $F(p_1, p_2, p_3) = F(p_1+p_2, p_3) + (p_1+p_2) F(\frac{p_1}{p_1+p_2}, \frac{p_2}{p_1+p_2})$.
    \item Group $(p_2, p_3)$ first: Similarly, we get
    $F(p_1, p_2, p_3) = F(p_1, p_2+p_3) + (p_2+p_3) F(\frac{p_2}{p_2+p_3}, \frac{p_3}{p_2+p_3})$.
\end{enumerate}
By letting $p_1 = \alpha\beta$, $p_2 = \alpha(1-\beta)$, and $p_3 = 1-\alpha$, and equating the two expressions, we can derive the fundamental functional equation for $\phi$:
\begin{equation}
    \phi(\alpha\beta) = \phi(\alpha) + \phi(\beta)
\end{equation}
This equation, known as the multiplicative Cauchy functional equation, dictates that the function must map products to sums.

\textbf{Step 2: The Logarithmic Solution.}
Under mild regularity conditions such as continuity or measurability, the only non-trivial solution to the equation $\phi(xy) = \phi(x) + \phi(y)$ for $x, y \in (0, 1)$ is the logarithmic function:
\begin{equation}
    \phi(\alpha) = -k \log \alpha
\end{equation}
for some constant $k$.

\textbf{Step 3: Generalization to M-dimensions.}
Any M-dimensional distribution $\mathbf{p} = (p_1, ..., p_M)$ can be constructed by recursively applying the decomposability rule. By repeatedly decomposing the distribution, we can show that the general form of $F$ must be:
\begin{equation}
    F(\mathbf{p}) = k \sum_{i=1}^{M} p_i \log(1/p_i) + B = k H(\mathbf{p}) + B'
\end{equation}
where $H(\mathbf{p})$ is the Shannon entropy and $k, B'$ are constants.

\textbf{Step 4: Calibration and Uniqueness.}
The remaining axioms fix the constants. The Transfer Principle (A4) implies that the function must be Schur-concave, which requires $k > 0$. Axiom A3 (Calibration) states that $F$ is minimized for the uniform distribution $\mathbf{u}_M = (1/M, ..., 1/M)$ and maximized at the vertices. This allows us to express the function as an "entropy gap" from the maximum possible entropy, $\log M$:
\begin{equation}
    F(\mathbf{p}) = k(\log M - H(\mathbf{p})) + C
\end{equation}
This is the unique form up to a positive affine transformation (scaling by $k$ and shifting by $C$). By applying a final normalization to map the minimum imbalance to 0 and maximum to 1, we uniquely arrive at the form used in this paper: $I(\mathbf{p}) = 1 - H(\mathbf{p})/\log_2 M$. This completes the proof.

\section{Proof of Lemma~\ref{lemma:invariance} (Invariance)}
\label{app:proof_lemma_invariance_new}
\begin{proof}
The proof follows directly from the fact that the exponent in the Sigmoid function, $a(h_i - h_0)$, remains unchanged under both transformations. For translation, the exponent becomes $a((h_i+c) - (h_0+c)) = a(h_i - h_0)$. For scaling, the exponent becomes $(a/\lambda)(\lambda h_i - \lambda h_0) = a(h_i - h_0)$. Since the exponent is invariant, all satisfaction scores $\{w_{u,v}\}$ are invariant. Consequently, the average satisfaction $\bar{s}$ and the entire probability distribution $\{p_{u,v}\}$ are invariant, which in turn leaves the Imbalance $I$ unchanged.
\end{proof}

\section{Proofs of Analytical Properties: Continuity and Differentiability}
\label{app:proof_analytical_props}

\begin{proposition}[Interior differentiability]
For any finite SLA parameters $(a,h_0)$ that result in all path satisfaction scores being strictly positive ($w_{u,v}>0$), the resulting probability vector $\mathbf{p}$ lies in the interior of the simplex, $\Delta^{M-1}$. In this case, the Imbalance function $I(\mathbf{p})$ is continuously differentiable ($C^1$).
\end{proposition}

\begin{proposition}[Boundary behavior and subgradients]
On the boundary of the simplex (where at least one $p_i=0$), the Shannon entropy $H$ is no longer differentiable but remains continuous and convex. Consequently, the Imbalance function $I$ is well-defined and admits subgradients.
\end{proposition}

\paragraph{Remark on the large-$a$ limit.}
The case of $a \to \infty$ is a special limit where the probability vector $\mathbf{p}$ approaches the boundary of the simplex, leading to the non-differentiability mentioned above. All gradient-based statements in the main text (e.g., the Covariance Rule) are asserted for finite $a$. In the limit, the analysis relies on the piecewise-constant nature of $I_\infty(h_0)$, whose derivatives are zero on the "steps" and undefined at the "risers," or more formally, can be described using subdifferentials.

These properties are based on the chain of function compositions: $(a, h_0) \xrightarrow{f_1} \{w_{u,v}\} \xrightarrow{f_2} \{\bar{s}, \{p_{u,v}\}\} \xrightarrow{f_3} I$.

\textbf{Proposition (Continuity).}
\textit{The metrics $I(a,h_0)$ and $\bar{s}(a,h_0)$ are continuous functions with respect to the SLA parameters $a$ and $h_0$.}
\begin{proof}
The proof demonstrates the continuity of each link in the composition chain.
\begin{enumerate}
    \item \textbf{Continuity of $w_{u,v}$ and $\bar{s}$:} The weight function $w(u,v; a, h_0)$ is a composition of elementary continuous functions and is therefore continuous. The average satisfaction $\bar{s}$, being a finite sum of these continuous functions, is also continuous.
    \item \textbf{Continuity of $p_{u,v}$:} The total weight $W = \sum w(u,v)$ is a continuous function. For any non-trivial graph, $W > 0$. The normalized probability $p_{u,v} = w_{u,v} / W$ is the quotient of two continuous functions with a non-zero denominator, and is thus continuous.
    \item \textbf{Continuity of $I$:} The Shannon entropy function $H$ is a continuous function of the probability vector $\mathbf{p}$. By the composite function continuity theorem, $H$ is a continuous function of $(a,h_0)$. Since $I = 1 - H/H_{\max}$ is a linear transformation of $H$, it is also continuous.
\end{enumerate}
\end{proof}

\textbf{Proposition (Differentiability).}
\textit{The metrics $I(a,h_0)$ and $\bar{s}(a,h_0)$ are differentiable functions of $a$ and $h_0$.}
\begin{proof}
The proof follows the same logic, demonstrating differentiability at each step.
\begin{enumerate}
    \item \textbf{Differentiability of $w_{u,v}$ and $\bar{s}$:} The Sigmoid function is infinitely differentiable ($C^\infty$). Therefore, $w(u,v; a, h_0)$ is differentiable, and so is the finite sum $\bar{s}$.
    \item \textbf{Differentiability of $p_{u,v}$:} Since $w_{u,v}$ and $W$ are differentiable, their quotient $p_{u,v}$ is also differentiable.
    \item \textbf{Differentiability of $I$:} The Shannon entropy function $H$ is differentiable. By the multivariate chain rule, the composite function $H(a,h_0)$ is differentiable. The Imbalance $I$, being a linear transformation of $H$, is therefore also differentiable.
\end{enumerate}
\end{proof}

\section{Proof of Theorem~\ref{thm:cov_grad_I_new} (Covariance Gradient Criterion)}
\label{app:proof_cov_grad}
\begin{proof}
The proof proceeds in three steps by applying the chain rule. We use subscript $i$ to denote a node pair $(u,v)$.
\textbf{Step 1:} Differentiate $I$ with respect to probabilities $p_i$.
The definition is $I = 1 - H / H_{\max}$. Using the natural logarithm for the derivation, $H = -\frac{1}{\ln 2} \sum_i p_i \ln(p_i)$. The derivative with respect to a parameter $\theta \in \{a, h_0\}$ is:
\begin{equation}
    \frac{\partial I}{\partial \theta} = \frac{1}{H_{\max} \ln 2} \sum_i \frac{\partial p_i}{\partial \theta} (1 + \ln p_i) \label{eq:app_proof_step1_full}
\end{equation}

\textbf{Step 2:} Differentiate probabilities $p_i$ with respect to weights $s_i$.
Given $p_i = s_i / W$ and the general form of the weight derivative $\partial s_i / \partial \theta = s_i \cdot g_i^{(\theta)}$, the quotient rule yields:
\begin{align}
    \frac{\partial p_i}{\partial \theta} = p_i (g_i^{(\theta)} - E_{\mathbf{p}}[g^{(\theta)}]) \label{eq:app_proof_step2_full}
\end{align}
where $E_{\mathbf{p}}[g^{(\theta)}]$ is the expectation of the sensitivity term $g^{(\theta)}$ over the distribution $\mathbf{p}$.

\textbf{Step 3:} Combine and identify the covariance structure.
Substituting Eq.~\eqref{eq:app_proof_step2_full} into Eq.~\eqref{eq:app_proof_step1_full} gives:
\begin{equation}
    \frac{\partial I}{\partial \theta} = \frac{1}{H_{\max} \ln 2} \sum_i p_i (1 + \ln p_i) (g_i^{(\theta)} - E_{\mathbf{p}}[g^{(\theta)}])
\end{equation}
The summation term is precisely the definition of the weighted covariance, $\mathrm{Cov}_{\mathbf{p}}(1 + \ln p, g^{(\theta)})$. This completes the proof.
\end{proof}

\section{Proof of Theorem~\ref{thm:small_a_I_new} (Small-$a$ Second-Order Law)}
\label{app:proof_small_a}

\begin{proposition}[Small-$a$ expansion with remainder]
Let $\mu=\mathbb{E}[h]$, $\sigma^2=\mathrm{Var}(h)$, $m_3=\mathbb{E}[|h-\mu|^3]$, and $m_4=\mathbb{E}[(h-\mu)^4]$ be the moments of the hop-count distribution. For sufficiently small $a$, the Imbalance can be expanded as:
\[
I(a,h_0)=\frac{a^2}{8\log_2 M}\,\sigma^2\;+\;R(a;h_0),
\]
with a remainder term $R(a;h_0)$ that is bounded by the higher-order moments of the path distribution. The leading-order term is independent of $h_0$.
\end{proposition}

\begin{proof}[Proof Sketch]
The proof involves a more detailed Taylor expansion of the satisfaction scores $w_i$ and the normalized probabilities $p_i$. By expanding the logarithm term in the entropy formula to the third order, the remainder term $R(a;h_0)$ can be shown to be controlled by the third and fourth central moments ($m_3$ and $m_4$) of the path distribution, with terms of order $O(a^3)$ and $O(a^4)$. The full derivation is straightforward but lengthy and is omitted for brevity.
\end{proof}

\begin{proof}
Let $x_i = a(h_i - h_0)$. As $a \to 0$, we use a Taylor expansion for the satisfaction score $s_i = (1+e^{x_i})^{-1}$ around $x=0$:
\begin{equation}
    s_i = \frac{1}{2} - \frac{a(h_i-h_0)}{4} + O(a^3)
\end{equation}
Summing over all $M = N(N-1)$ pairs gives the total weight $W$. The normalized probability $p_i = s_i / W$ is then approximated as:
\begin{equation}
    p_i = \frac{1}{M} \left(1 + \frac{a(\mu_h - h_i)}{2} \right) + O(a^2)
\end{equation}
This shows $p_i$ is a small perturbation from the uniform distribution $\mathbf{u}$. The Shannon entropy for such a distribution has a standard second-order expansion. By calculating the sum of squared perturbations, which is proportional to $a^2 \mathrm{Var}(h)$, we find the entropy to be:
\begin{equation}
    H(\mathbf{p}) \approx \log_2(M) - \frac{a^2 \mathrm{Var}(h)}{8 \ln 2}
\end{equation}
Substituting this into the definition of Imbalance $I = 1 - H/H_{\max}$ yields the final result.
\end{proof}

\section{Proof of Theorem~\ref{thm:large_a_I_new} (Large-$a$ Step-Function Limit)}
\label{app:proof_large_a}
\begin{proof}
The proof relies on the pointwise limit of the Sigmoid function as $a \to \infty$. Let $K(h_0)$ be the number of pairs with path cost $h_i < h_0$.
\begin{itemize}
    \item If $h_i < h_0$, the exponent $a(h_i-h_0) \to -\infty$, so $s_i \to 1$.
    \item If $h_i > h_0$, the exponent $a(h_i-h_0) \to +\infty$, so $s_i \to 0$.
\end{itemize}
The total weight $W = \sum s_i$ converges to the count $K(h_0)$. The probability distribution $\mathbf{p}_\infty$ thus becomes a uniform distribution over the $K(h_0)$ pairs that satisfy the threshold, and zero elsewhere. The entropy of such a uniform distribution is simply $H(\mathbf{p}_\infty) = \log_2(K(h_0))$. Substituting this into the definition of $I$ gives the final result.
\end{proof}

\begin{proof}
The proof establishes a bound on the difference between the actual Imbalance $I(a,h_0)$ and its limit $I_\infty(h_0)$ by bounding the L1-norm distance between their underlying probability distributions, $\mathbf{p}(a)$ and $\mathbf{p}(\infty)$. Let $\delta = \min_{i: h_i \neq h_0} |h_i - h_0|$ be the minimum distance from the threshold to any path cost.

\paragraph{Step 1: Pointwise Weight Error.}
The error for a single satisfaction score $w_i(a)$ from its limit $w_i(\infty) \in \{0, 1\}$ can be bounded exponentially:
\[
|w_i(a)-w_i(\infty)|\ \le\ e^{-a|h_i-h_0|}\ \le\ e^{-a\delta}.
\]

\paragraph{Step 2: Total Variation Bound on Probabilities.}
The L1-norm distance (total variation distance) between the probability vectors $\mathbf{p}(a)$ and $\mathbf{p}(\infty)$ can be bounded in terms of the pointwise weight errors. A standard result provides the bound:
\[
\|\mathbf{p}(a)-\mathbf{p}(\infty)\|_1 \;\le\; 
\frac{2\sum_i |w_i(a)-w_i(\infty)|}{\sum_i w_i(\infty)}
\;\le\; \frac{2 M e^{-a\delta}}{K(h_0)},
\]
where $K(h_0)$ is the number of paths with $h_i < h_0$. Let us denote this exponential error bound as $\varepsilon(a, \delta) \triangleq C e^{-a\delta}$ for a constant $C$ that depends on the topology.

\paragraph{Step 3: Entropy Continuity via Fannes–Audenaert Inequality.}
The Fannes–Audenaert inequality provides a tight bound on the difference between the entropies of two probability distributions based on their L1-norm distance. A simplified form states that for $\epsilon = \frac{1}{2}\|\mathbf{p}-\mathbf{q}\|_1 \le \frac{1}{2}$,
\[
|H(\mathbf{p})-H(\mathbf{q})| \ \le\ \epsilon\log_2(M-1)+H_b(\epsilon),
\]
where $H_b$ is the binary entropy function. For small $\epsilon$, this implies $|H(\mathbf{p})-H(\mathbf{q})| \approx O(\epsilon \log M)$. Applying this to our probability vectors, the error in the Imbalance metric is:
\[
|I(a,h_0)-I_\infty(h_0)| = \frac{|H(\mathbf{p}(a))-H(\mathbf{p}(\infty))|}{\log_2 M} \le C' \varepsilon(a, \delta),
\]
which demonstrates the asserted exponential decay. This implies that for the imbalance value to transition between steps, the threshold $h_0$ must move across a path cost, and the width of this transition region is of order $\Delta h_0=\Theta(1/a)$.
\end{proof}

\section{Robustness of the Large-$a$ Limit Validation}
\label{app:large_a_robustness}
To demonstrate that the choice of $a=10.0$ in our main validation was not cherry-picked, this appendix provides a supplementary validation for a range of large $a$ values. The theoretical foundation for this lies in Proposition~\ref{prop:width}, which states that the convergence to the large-$a$ limit is exponentially fast.

Figure~\ref{fig:large_a_supplementary} visually confirms this property. The plot compares the theoretical staircase function, $I_\infty(h_0)$ (black dashed line), with the simulated curves for a range of strictness parameters ($a=5.0, 10.0, 20.0$). As the value of $a$ increases, the simulated curve gets progressively sharper and hugs the theoretical staircase more closely. This demonstrates that the large-$a$ limit is not an artifact of a single parameter choice, but rather a robust property of the framework that is approached exponentially as the SLA becomes more stringent. The excellent agreement for $a=20.0$ shows that even for a finite strictness, the theoretical limit provides a highly accurate prediction of the system's behavior.

\begin{figure}[h!]
    \centering
    \includegraphics[width=0.95\columnwidth]{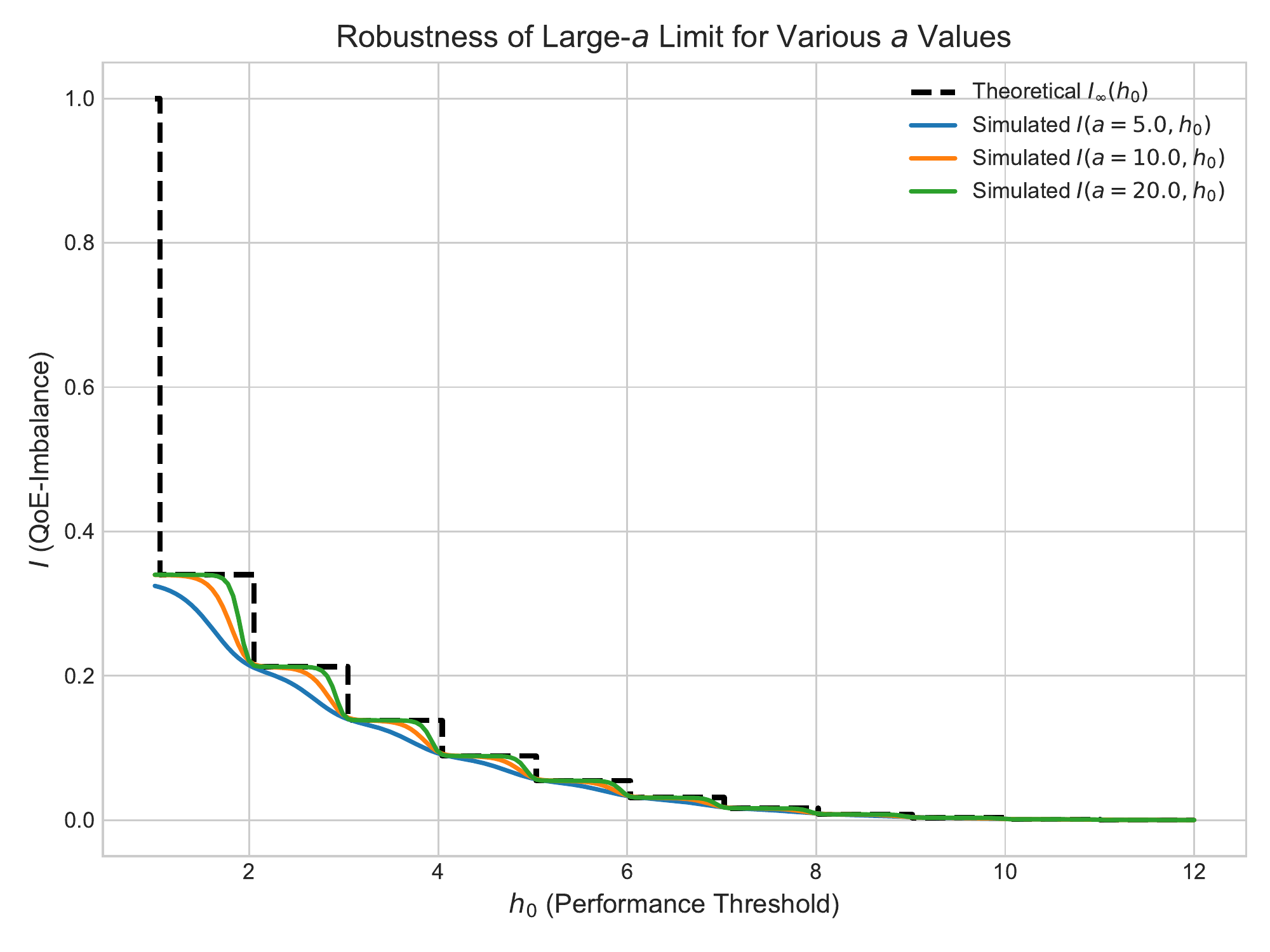}
    \caption{Supplementary validation of the Large-$a$ limit for multiple values of $a$. All simulated curves rapidly converge to the theoretical limit, confirming the robustness of the validation.}
    \label{fig:large_a_supplementary}
\end{figure}


\bibliographystyle{IEEEtran}
\bibliography{main}

\end{document}